\documentclass{amsart}

\usepackage{amssymb,amsmath,amsthm,amsfonts}

\usepackage[
  paperwidth = 8.5in,   
  paperheight = 11in,
  left   = 1in,         
  right  = 1in,        
  top    = 1in,
  bottom = 1in
]{geometry}

\usepackage{graphicx}
\usepackage{enumitem}
\usepackage{caption}
\usepackage{lmodern}
\usepackage[font=small, labelfont=bf]{subcaption}
\usepackage[font=small, labelfont=bf]{caption}

\usepackage{xcolor}
\definecolor{royalpurple}{RGB}{99, 6, 161} 

\newtheorem{theorem}{Theorem}
\newtheorem{lemma}{Lemma}
\newtheorem{proposition}{Proposition}

\usepackage{algorithm}
\usepackage{algpseudocode}

\usepackage{hyperref}
\hypersetup{
    colorlinks=true,      
    linkcolor=blue,       
    citecolor=royalpurple,       
    urlcolor=blue        
}

\begin{document}

\title[]{An Efficient Genus Algorithm \\ Based on Graph Rotations}

\author[]{Alexander Metzger}
\address{University of Washington}
\email{metzgera@uw.edu}

\author[]{Austin Ulrigg}
\address{University of Washington}
\email{austinul@uw.edu}


\commby{Insert Editor's name} 

\begin{abstract}
    We study the problem of determining the minimal \textit{genus} of a simple finite connected graph. We present an algorithm which, for an arbitrary graph $G$ with $n$ vertices and $m$ edges, determines the orientable genus of $G$ in  $\mathcal{O}(n(4^m/n)^{n/t})$ steps where $t$ is the \textit{girth} of $G$. This algorithm avoids difficulties that many other genus algorithms have with handling bridge placements which is a well-known issue \cite{myrvold2011}. Its implementation can be found \href{https://github.com/SanderGi/Genus}{here} under an MIT license. The algorithm has a number of useful properties for practical use: it is simple to implement, it outputs the faces of an optimal embedding, and it iteratively narrows both upper and lower bounds. We illustrate the algorithm by determining the genus of the $(3,12)$ cage (which is 17); other graphs are also considered.
\end{abstract}
\vspace{-3em}
\maketitle
\vspace{-3em}
\section{Introduction}
\subsection{Motivation and Background} Say that you have three houses and three utilities, and you must connect each house to each utility via a wire, is there a way to do this so that none of the wires cross each other? This problem can be reframed in terms of graph theory: is $K_{3,3}$ planar? Kuratowski's theorem \cite{kuratowski1930}  tells us that it is not. However, $K_{3,3}$ is \textit{toroidal}, it can be embedded on a torus without any edges crossing. 

\begin{figure}[H]
    \centering
    \begin{subfigure}[t]{0.4\textwidth}
        \centering
        \includegraphics[height=0.8in]{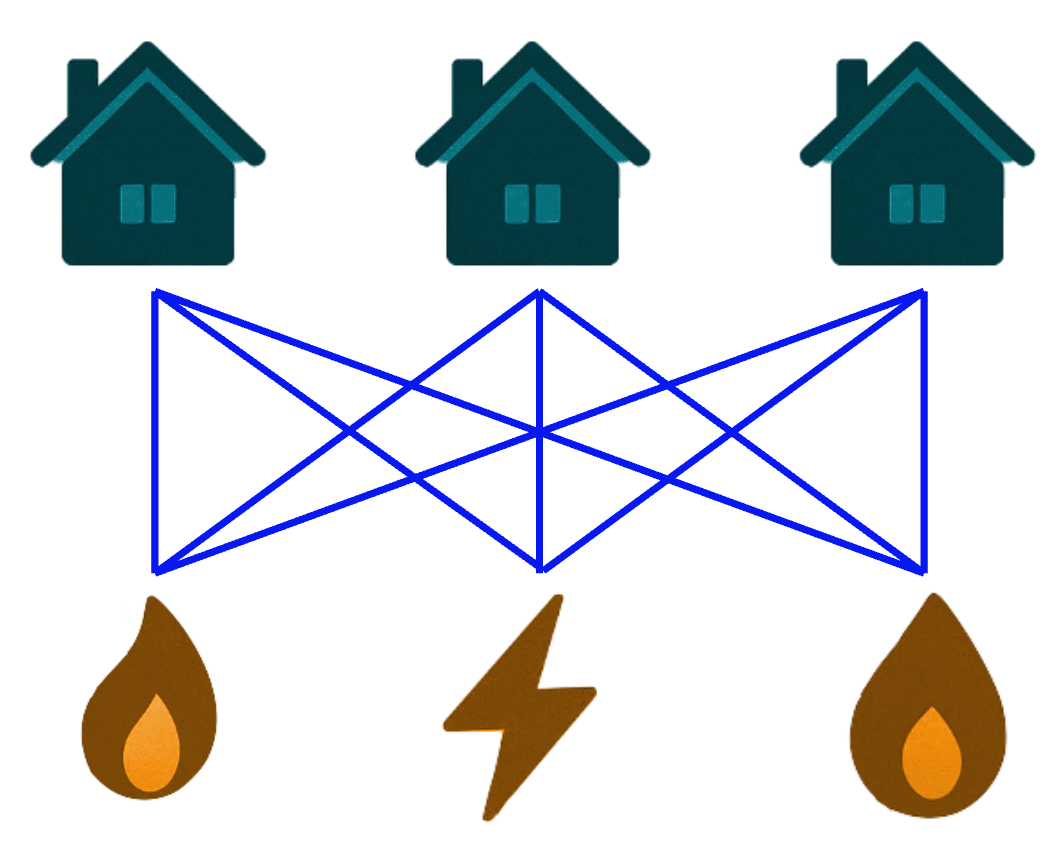}
        \caption{The Complete Bipartite Graph $K_{3,3}$.}
    \end{subfigure}
    \begin{subfigure}[t]{0.4\textwidth}
        \centering
        \includegraphics[height=0.8in]{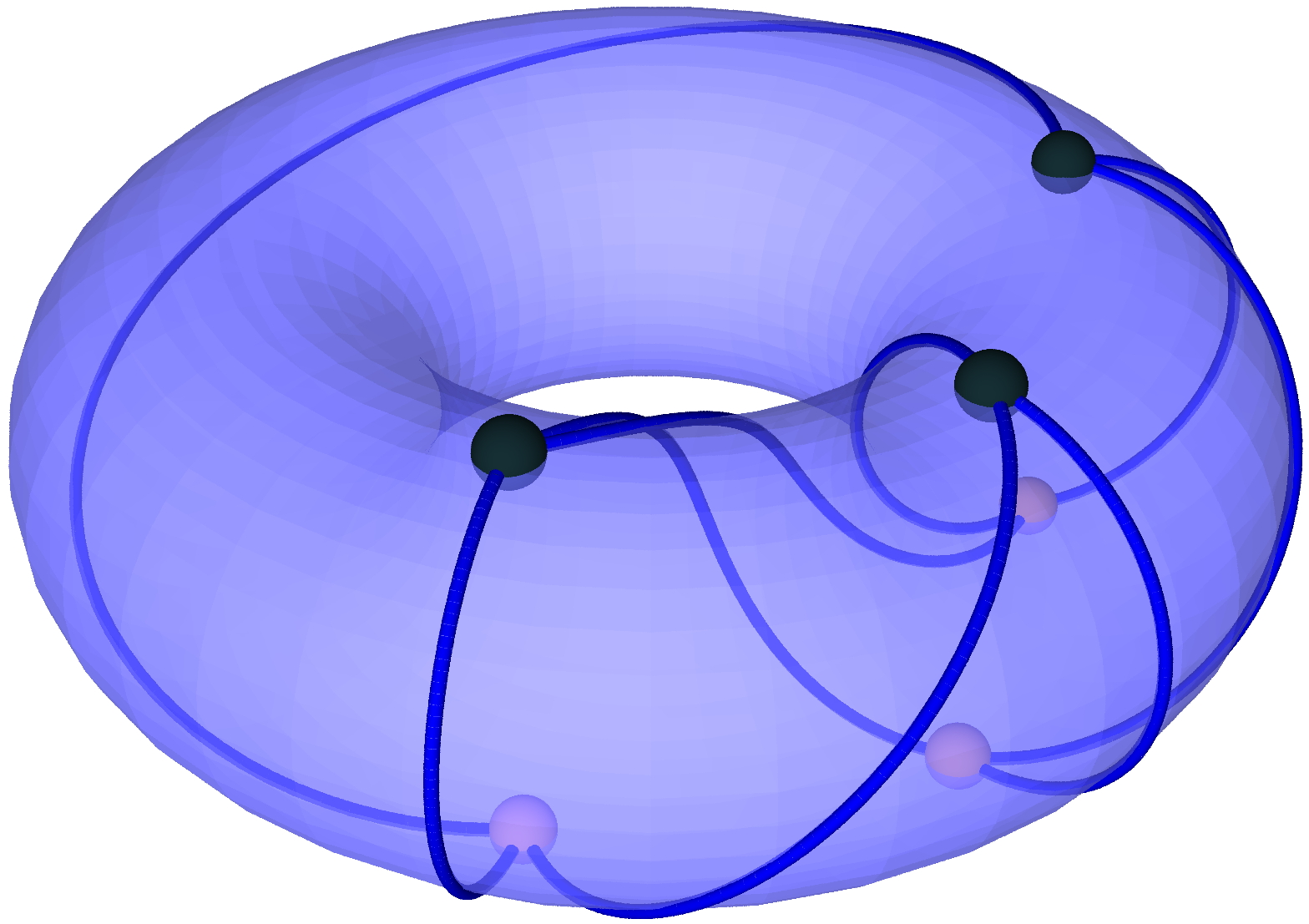}
        \caption{$K_{3,3}$ Embedded on a Torus.}
        \label{fig:k33embedding}
    \end{subfigure}
\end{figure}

The property of a torus that allows us to embed $K_{3,3}$ is that it has a handle (unlike surfaces such as a plane or a sphere). This motivates classifying surfaces by their number of handles, that is, their genus $g$. In these terms, we have seen that the minimum genus surface that $K_{3,3}$ can be embedded has $g=1$, and we say that $K_{3,3}$'s genus is $1$. Although one may also consider embeddings on non-orientable surfaces, the focus of our algorithm in this paper is only on the orientable genus. In the orientable case for genus zero we use the special name ``planar" and for genus one we use ``toroidal". Similarly, it is known that the complete graph with 7 vertices, $K_7$, has genus $1$ and can be embedded on a torus. However, $K_8$ cannot be embedded on a torus, and has genus $2$. In fact, Ringel \cite{ringel1954,ringel1965}, determined the minimum non-orientable genus for the complete graph $K_n$ and also the orientable and non-orientable genus for the complete bipartite graph $K_{m,n}$. Further, Ringels and Youngs later determined the minimum orientable genus for $K_n$ \cite{ringel1968}:
\[
    g(K_n) = \left\lceil \frac{(n-3)(n-4)}{12} \right\rceil \qquad g(K_{m,n}) = \left\lceil \frac{(n-2)(m-2)}{4} \right\rceil
\]
However, it is not always so simple to determine the genus of an arbitrary graph. For example, the following are examples of graphs with unknown genus. 

\begin{figure}[H]
    \centering
    \begin{subfigure}[t]{0.4\textwidth}
        \centering
        \includegraphics[width=0.4\linewidth]{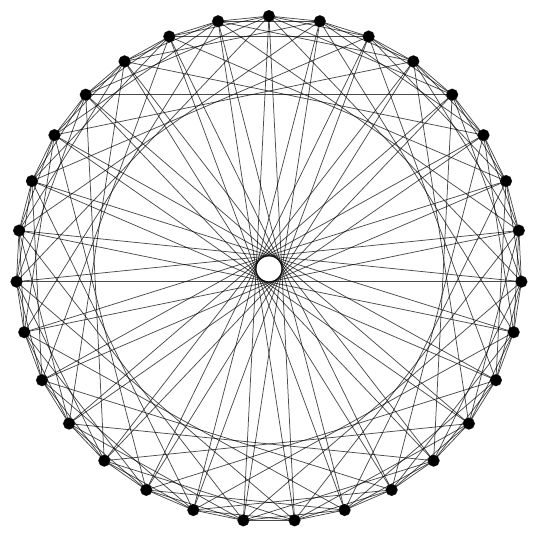}
        \caption{Cyclotomic 31 Graph ($12\leq g \leq 26$).}
    \end{subfigure}
    \begin{subfigure}[t]{0.4\textwidth}
        \centering
        \includegraphics[width=0.4\linewidth]{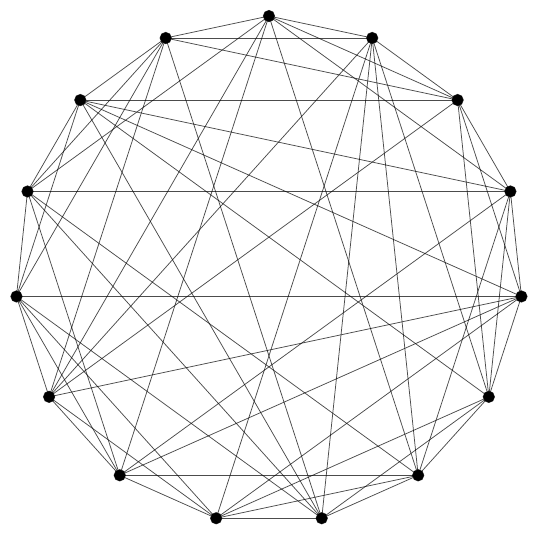}
        \caption{Johnson (6,2) Graph ($4\leq g \leq 5$).} 
    \end{subfigure}
\end{figure}

 The challenge of determining the orientable genus of graphs and constructing their embeddings is a fundamental problem in graph theory, with applications in map colouring, very large scale integration, topology, and network science.

\section{Main Result}
 Throughout this paper $G$ denotes a finite connected simple graph, with $n$ vertices, $m$ edges, and girth $t$. We present a simple algorithm to determine the genus of a graph: \textsc{Practical\_Algorithm\_for\_Graph\_Embedding} (\textsc{PAGE}). The algorithm runs faster than previously implemented algorithms including those presented in \cite{beyer2016, chimani2019, sagemath2024}. \textsc{PAGE} can easily handle graphs like $K_7$ and $K_8$ in a few seconds and scales well to graphs with over a hundred edges, which it can run in a few minutes. The algorithm also provides upper and lower bounds which it iteratively narrows as it runs, allowing for practical applications. 
\begin{theorem}[Main Result]
\textsc{PAGE} described in \S \ref{sec:Construction} determines the genus of $G$ with runtime of $\mathcal{O}(n(4^m/n)^{n/t})$.
\end{theorem}
 
 We emphasize that \textsc{PAGE} is relatively easy to implement; moreover \S \ref{sec:examples} contains a number of concrete examples where the algorithm is used to determine the genus of graphs whose genus was previously unknown.
 
\section{Preliminaries}
\label{sec:prelim}
For the convenience of the reader, following along with \textit{Graphs on Surfaces} \cite{mohar2001graphs} by Mohar and Thomassen, we provide the following
definitions and relevant theorems.

\subsection{Graph Embeddings and Surfaces}

Intuitively, one can think of an embedding of a graph on a surface as a drawing of the graph on that surface without any edges crossing. We will now make this precise. A \textit{curve} in a topological space $X$ is the image of a continuous map $f:[0,1]\to X$. It is called \textit{simple} if $f$ is injective, and a simple arc with endpoints $f(0)$ and $f(1)$ is said to \textit{connect} these endpoints \cite[Sec.~2.1]{mohar2001graphs}. A graph $G$ is \textit{embedded} in a topological space $X$ if the vertices of $G$ are distinct points in $X$ and every edge of $G$ is represented by a simple arc connecting in $X$ the vertices it joins in $G$, such that the interior of each arc is disjoint from other edges and vertices \cite[Sec.~2.1]{mohar2001graphs}. When we later discuss rotation systems and facial walks, we use the associated directed edges: an undirected edge joining vertices $u$ and $v$ has two directed versions, $(u,v)$ and $(v,u)$, corresponding to the two possible directions of traversal along the same arc. An \textit{embedding} of $G$ into $X$ is thus an isomorphism of $G$ with a graph embedded in $X$. If such an embedding exists, we say that $G$ \textit{can be embedded in} $X$ \cite[Sec.~2.1]{mohar2001graphs}. An embedding of $G$ in a surface $S$ is said to be \textit{cellular} (or a \textit{2-cell embedding}) if every face of $G$ is homeomorphic to an open disc in $\mathbb{R}^2$. Every cellular embedding is a 2-cell embedding, and vice versa \cite[Thm.~3.2.4]{mohar2001graphs}.

\subsection{Rotation Systems}

Cellular embeddings can be encoded combinatorially by listing the local rotations at each vertex. Given a cellular embedding of $G$ in a surface $S$, let $\Pi = \{\pi_v \mid v \in V(G)\}$, where each $\pi_v$ is a cyclic permutation of the edges incident with vertex $v$, such that $\pi_v(e)$ is the successor of $e$ in the clockwise ordering around $v$. Each permutation $\pi_v$ is called the \textit{local rotation} at vertex $v$, and the set $\Pi$ itself is the \textit{rotation system} of the embedding of $G$ in $S$ \cite[p.~90]{mohar2001graphs}. 

\begin{figure}[h!]
    \centering
\includegraphics[width=0.35\textwidth]{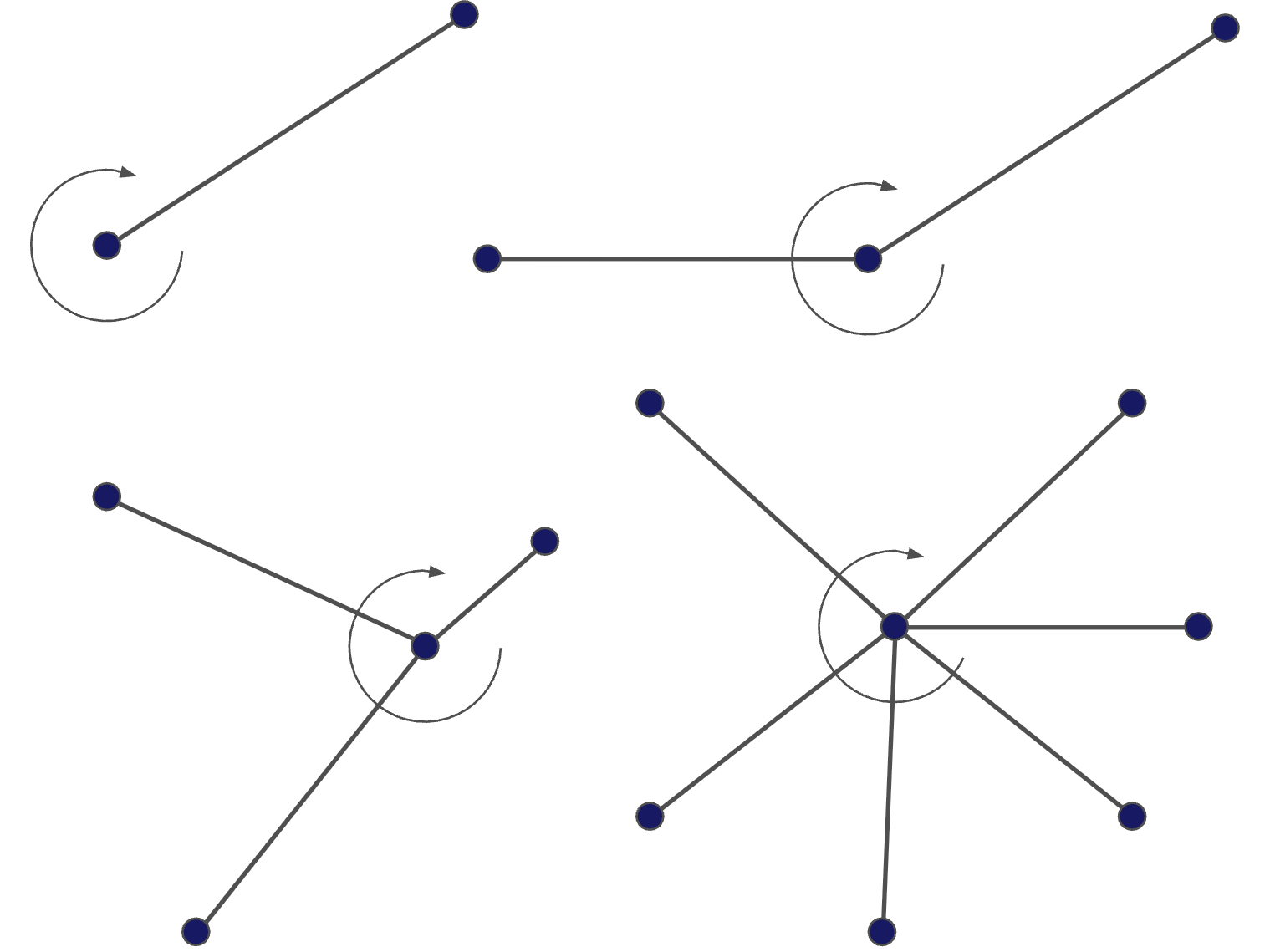}
    \caption{Local Rotations at a Vertex.}
    \label{fig:rotationsystem}
\end{figure}

Rotation systems determine the embedding up to homeomorphism, as shown below.
\begin{theorem}[Heffter-Edmonds-Ringel]\cite[p.~90]{mohar2001graphs}
    Suppose that $ G $ is a connected multigraph with at least one edge that is cellularly embedded in an orientable surface $ S $. Let $ \Pi $ be the rotation system of this embedding, and let $ S' $ be the surface of the corresponding 2-cell embedding of $ G $. Then there is a homeomorphism of $ S $ onto $ S' $ taking $ G $ in $S$ onto $G$ in $S'$ (in such a way that we induce the identity map from $G$ onto its copy in $S'$). In particular, every cellular embedding of a graph $ G $ in an orientable surface is uniquely determined, up to homeomorphism, by its rotation system.
\end{theorem}

\subsection{Orientable Genus}

The orientable genus of a graph $G$ quantifies exactly how many handles a surface needs for $G$ to be embedded in it. The \textit{orientable genus} $g$ of a graph $ G $ is the smallest integer $ h $ for which $ G $ embeds in the orientable surface $ S_h $, the connected sum of $h$ tori \cite[p.~94]{mohar2001graphs}.

\subsection{Facial Walks of a Rotation System}
Once $\Pi=\{\pi_v\}$ is fixed, we build each face by starting at a vertex edge pair and traversing as follows:
\[
  (v,e) \xrightarrow{\;\pi_v\;} (v,\,\pi_v(e))
  \xrightarrow{\;\text{traverse}\;} (v',\,\pi_v(e))
  \xrightarrow{\;\pi_{v'}\;} \cdots
\]
Because $G$ is finite this returns to the start, tracing out a closed walk in $G$. Such a closed walk is called a \emph{facial walk} of~$\Pi$. The amount of distinct facial walks of a rotation system $\Pi$ is denoted $F$ (we do not distinguish a facial walk from any of its cyclic shifts). An embedding attaining the minimal genus is a \textit{minimum genus embedding} \cite[p.~95]{mohar2001graphs}. Every minimum genus embedding of a connected graph is cellular \cite[Prop.~3.4.1]{mohar2001graphs}. Hence, genus computations may focus exclusively on cellular embeddings, and in particular, by Theorem 2 above, only on rotation systems. In fact,
\begin{theorem}[Every Rotation System Corresponds to an Embedding]\cite[p.~87]{mohar2001graphs}
    Every connected multigraph with at least one edge admits a 2-cell embedding in some orientable surface. The embedding is constructed by forming a polygon for each facial walk defined by the rotation system, and identifying the sides of these polygons along their shared edges.
\end{theorem}
Thus, when enumerating possible embeddings of $ G $, one may safely iterate over all $ \prod_{v\in V}(\deg(v)-1)! $ rotation systems, since each corresponds to an embedding on an orientable surface of some genus. The genus $g$ of the surface can then be computed using Euler’s formula,
$$n-m+F=2-2g, \implies g=\frac{-n+m-F+2}{2}$$
Thus, the smallest $g$ occurs when $F$ is maximal, i.e, when we find a rotation system with the most amount of induced facial walks. 

\subsection{Definitions}
Throughout this paper, a \textit{non-backtracking closed directed trail} is a walk that starts and ends at the same vertex, uses no repeated directed edges, and does not immediately follow any directed edge by its reverse. In this paper, a \textit{closed trail} always refers to a non-backtracking closed directed trail unless otherwise stated. Additionally, we say that a collection $\mathcal{F}$ of \textit{closed trails} is \textit{realizable} as a rotation system of $G$ if there exists a rotation system $\Pi$ of $G$ whose facial walks exactly match the closed trails in $\mathcal{F}$.

\subsection{Prior Work} It has been well established that the problem of determining the genus is NP-hard \cite{thomassen1989}. However, it is tractable for fixed genus \cite{mohar1999,myrvold2011, robertson1995}. Mohar has proven the existence of a linear time algorithm for arbitrary fixed surfaces \cite{mohar1999}. Unfortunately, these theoretical results are non-constructive and so far no one has been able to design an algorithm that achieves such performance \cite{myrvold2011}. For instance, Robertson and Seymour showed that checking if a graph contains a given minor can be done in cubic time \cite{robertson1995} and that, for each fixed surface, there are only finitely many forbidden minors that determine whether a graph embeds in that surface \cite{robertson2004}. However, the complete list of forbidden minors is only known for planar \cite{kuratowski1930} and projective plane graphs \cite{archdeacon1981}. Even small toroidal graphs (genus 1) cannot be solved with this approach since there are at least 17.5K toroidal minors and likely many more \cite{myrvold2018}. In practice, the finite number of graph minors scales at an impractical exponential rate with the genus. This makes it intractable to compute all the minors \cite{myrvold2011}. As it stands, the best implemented algorithms are exponential time, and the rest are too complex/impractical for implementation and have intractable constant factors \cite{chambers2002, myrvold2011, myrvold2018}.

\subsection{Existing Algorithms and Limitations}
An algorithm called \textsc{multi\_genus} by Gunnar Brinkmann is a particularly fast method for computing genus of graphs with relatively low genus compared to the vertex degree \cite{brinkmann2022}. Although a formal runtime analysis of \textsc{multi\_genus} has not been presented, experimental results suggest it handles graphs with higher vertex degrees efficiently. However, \textsc{PAGE} seems to remain advantageous for graphs with vertices only of degree 5 or lower. Additionally, our approach is comparatively simpler to implement. As it stands, \textsc{multi\_genus} represents the fastest known approach for high-degree graphs, and low relative genus, whereas \textsc{PAGE} provides an effective alternative, and is particularly advantageous for graphs with bounded vertex degree or high girth. 

Several other genus algorithms outside of \textsc{multi\_genus} have been developed, improving upon earlier methods through more optimized implementations and by using existing optimized solvers such as those for integer linear programming
\cite{beyer2016, chimani2019, sagemath2024}. They can easily compute the genus of graphs the size of $K_6$ in less than a second but, even just adding another vertex, $K_7$ could take hundreds of hours \cite{beyer2016}. $K_8$ and above is almost entirely out of reach. They omit formal runtime bounds, which are likely super-exponential for general graphs. For example, the best available open-source implementation in SageMath exhibits a worst-case complexity of $\mathcal{O}(n(n-1)!^{n})$ \cite{sagemath2024}.

\section{Results on Specific Graphs}
\label{sec:examples}

The purpose of this section is to describe various results that we obtained when using \textsc{PAGE} to determine the genus of certain graphs.

\subsection{Cages} \label{sec:examples_cages}
A graph family of special interest is the $(r, t)$-cage graphs. They are the smallest $r$-regular graphs with girth $t$. The genus of $(3, t)$ cage graphs is known up to $t = 10$, and \textsc{PAGE} extends these results by determining the genus of the $(3, 12)$ cage. Although the structure of $(3, t)$ cages is not fully known for $t > 12$, our approach is likely applicable to higher values of $t$ as these cages are discovered. We outperform all existing algorithms, including \textsc{multi\_genus}, for $t > 8$, and have the only tractable algorithm for $t \geq 12$.

\begin{theorem}
The genus of the unique $(3, 12)$ cage graph is 17. 
 \end{theorem}
An example of a huge graph that is impractical to embed optimally is the (6, 12) cage graph. It consists of 7812 vertices, 23436 edges, and an automorphism group of nearly 6 billion elements. Nonetheless, \textsc{PAGE} can still progressively narrow down the genus range. We established bounds for the genus of the (6,12) cage graph between 5860 and 7810 before encountering memory limitations. 
 
\subsection{Circulant and Complete Multipartite Graphs}

Another graph family that is of special interest is the complete $n$-partite graph $K_{{2, 2, \dots, 2}}$ ($n$ copies of 2), also known as the cocktail party graph of order $n$. It is conjectured to have genus $\lceil (n-1)(n-3)/3\rceil $ for all $n$, proven for all $n$ not a multiple of 3 \cite{jungerman1978}. The complete $n$-partite graph represents the problem of how many handshakes are possible in a room of $n$ couples if no one shakes their own partner's hand and has many applications in combinatorics. It is known that $K_{2, 2, \dots, 2}$ is isomorphic to the circulant graph $Ci_{2n}(1,2,\dots,n-1)$, defined as the graph with vertices labeled $0,1,\dots,2n-1$ where each vertex $i$ is adjacent to vertices $(i+j)\bmod 2n$ and $(i-j)\bmod 2n$ for every $j \in \{1,2,\dots,n-1\}$. The genus is also known for all circulant graphs with genus 1 and 2 \cite{conder2015}. However, not all circulant graphs are complete $n$-partite graphs or of small genus. In the vast majority of cases, the genus of arbitrary circulant graphs is unknown. Using \textsc{PAGE}, we were able to determine the genus for several circulant graphs where the values were previously unknown in less than a second:
\begin{theorem}[Circulants] The genus of $Ci_{14}(1,2,3,6)$ is 4. The genus of $Ci_{18}(1, 3, 9)$ is 4. The genus of $Ci_{20}(1, 3, 5)$ is 6. The genus of $Ci_{20}(1, 6, 9)$ is 6. \end{theorem}
Moreover, our approach also verified the genus of certain circulant graphs that correspond to well-known structures, such as the complete $n$-partite graphs: the genera of certain complete multipartite graphs are well-established: the genus of $K_{2,2}$ is $0$. The genus of $K_{2,2,2}$ is $0$. The genus of $K_{2,2,2,2}$ is $1$. The genus of $K_{2,2,2,2,2}$ is $3$. These values were verified with \textsc{PAGE}, consistent with the known results in the literature (see \cite{jungerman1978}).

\subsection{The Gray Graph}
Additionally, significant interest has surrounded the genus of the Gray graph (it happens to be 7), which has been addressed in a dedicated study \cite{gray2005}. Most existing algorithms, except \textsc{multi\_genus}, require over 42 hours to compute the genus of similar graphs, whereas \textsc{PAGE} achieves the same result in just a few minutes. 

\subsection{Progressive Refinement of Genus Bounds}
\label{sec:example_narrowing_bounds}
For large graphs, an exact genus is often infeasible and unnecessary in practice. For cases where it suffices to have an embedding within some error tolerance of the fewest handles, \textsc{PAGE} outputs genus bounds that narrow with increasing iterations.
\begin{theorem}
\label{theorem:4}
 For any graph $G$ with genus $g$, \textsc{PAGE} computes two monotone sequences of integers 
\[
      g_0 \leq g_1 \leq \dots \leq g_r=g,
     \qquad
     G_0 \geq G_1 \geq \dots \geq G_s = g,
\]
that converge to $g$ and,
\begin{enumerate}
\item \textsc{PAGE} halts when $g_i = G_j$, at which point $g_i=g=G_j$. 
\item  The length of both sequences is finite, so \textsc{PAGE} always terminates.
\end{enumerate}
\end{theorem}
The proof of Theorem \ref{theorem:4} is included later in \S\ref{proof:4}. 
To demonstrate the usefulness of Theorem 5, we applied \textsc{PAGE} to large graphs, establishing the genus bounds in Table \ref{tab:narrowing_bounds} within 15 minutes.\vspace{0.55em}

\begin{table}[H]
\centering
\footnotesize
\begin{tabular}{|c|c|c|c|}
\hline
 Graph Name & Reference & Genus Lower Bound & Genus Upper Bound  \\
\hline
Bipartite Kneser Graph (12, 3) & \cite{weisstein:bipkneser:2024c} & 4401 & 8979 \\
DifferenceSetIncidence(40, 13, 4) & \protect\textsuperscript{*} & 91 & 214 \\
Johnson Graph (8,4) & \cite{weisstein:johnson:2024d} & 60 & 238 \\
Johnson Graph (9, 4) & \cite{weisstein:johnson:2024d} & 148 & 558 \\
HoffmanSingletonBipartDoubleGraph & \protect\textsuperscript{*} & 68 & 100 \\
Higman Sims Graph & \cite{weisstein:hsgraph:2024a} & 226 & 490 \\
Cyclotomic Graph 61 & \cite{weisstein:cyclograph:2024b} & 73 & 265 \\
Cyclotomic Graph 67 & \cite{weisstein:cyclograph:2024b} & 91 & 325 \\
\hline
\end{tabular}
\caption{Genus bounds for large graphs established by \textsc{PAGE} within 15 minutes. Graph names follow the built-in naming used by \textit{Mathematica}. \textsuperscript{*}These graphs are defined internally by \texttt{GraphData}, but lack publicly available documentation.}
\label{tab:narrowing_bounds}
\end{table}

\section{Theoretical Justification of PAGE}
\label{sec:TheoryPAGE}
\subsection{Pre-Processing}
We will limit our analysis to connected graphs with vertices of degree $> 2$ since all graphs simplify to this case per Lemma \ref{lemma:connectedVertexGEQ1}.

\begin{lemma}
    \label{lemma:connectedVertexGEQ1}
    We can simplify any graph to remove vertices of degree $\leq 2$ without changing its genus. The genus of a disconnected graph can be calculated by summing the genera of its connected components.
\end{lemma}
\begin{proof}
    Degree $0$ vertices are isolated and can be drawn anywhere on the surface without causing edge crossings since they are not connected to any edges. Degree $1$ vertices likewise can always be added back into an embedding. Degree $2$ vertices can simply be removed and replaced with an edge connecting its neighbors. The additivity of the minimum orientable genus over connected components has been well established \cite{mohar2001graphs}.
\end{proof}
\subsection{Rotation Systems and Euler's Formula}
The general idea of \textsc{PAGE} is Euler's formula $n - m + F = 2 - 2g$ which links the number of facial walks $F$ with the genus $g$ \cite{euler1758}. Naively, finding the minimum genus amounts to searching all combinations of closed trails to find the one that corresponds to a maximal number of facial walks as seen in Lemma \ref{lemma:facialsubsetcycles}. Traditional improved exhaustive search algorithms instead search through the rotation systems since they each induce an embedding \cite{sagemath2024, heffter1891}. Searching through rotation systems however does not easily allow further pruning of the search space nor inform a heuristic search through rotation systems in an order that most quickly narrows down the genus. Our algorithm instead searches through closed trail combinations which facilitates a number of optimizations (early stopping, heuristic search) and, with our main contribution, still allows us to prune collections of closed trails that cannot be realized as a rotation system of $G$, which results in an exponentially reduced search space.
\subsection{Necessary Conditions for Realizing Closed Trails}
The following lemmas detail necessary conditions for realizing collections of closed trails as rotation systems, allowing \textsc{PAGE} to reduce the search space of all combinations of closed trails.
\begin{lemma}
    \label{lemma:directedEdgesUsed}
    Given a graph $G$, any set of facial walks induced by a rotation system must use each directed edge exactly once.
\end{lemma}
\begin{proof}
This is clear by the definition of a rotation system. 
\end{proof}

\begin{lemma} \label{lemma:facialsubsetcycles} Given a graph $G$, the set of closed trails $C$ of $G$ is finite and any set of facial walks induced by a single rotation system of $G$ is a subset of $C$. \end{lemma}
\begin{proof} The finite size of $C$ follows from the finite combinations of up to $2m$ directed edges and that a closed trail cannot repeat a directed edge and must therefore contain at most $2m$ directed edges. By Lemma \ref{lemma:directedEdgesUsed}, any facial walk in a set of facial walks induced by a rotation system of $G$ is a closed trail. \end{proof}

\begin{lemma} \label{lemma:vertexuses} Given a graph $G$ choose any vertex $v \in V$ of degree $d$. Then in any set of facial walks induced by a rotation system of $G$, the vertex $v$ occurs exactly $d$ times as part of a closed trail. \end{lemma}
\begin{proof} By Lemma \ref{lemma:directedEdgesUsed}, any set of facial walks induced by a rotation system of $G$ must use each directed edge of $G$ exactly once. Since each facial walk is a closed trail by Lemma \ref{lemma:facialsubsetcycles}, it must use $2k$ of the directed edges that touch $v$ for some non-negative integer $k$ representing the number of occurences of $v$ in that facial walk. There are $2d$ directed edges that touch $v$ and thus $2d / 2 = d$ occurrences of $v$. \end{proof}

\begin{lemma} \label{lemma:cycledistribution} Given a graph $G$, any set of facial walks induced by a rotation system must be a set of closed trails whose lengths add up to the number of directed edges $2m$. \end{lemma}
\begin{proof} By Lemma \ref{lemma:facialsubsetcycles}, the set of facial walks is a set of closed trails. The length of each trail is the number of directed edges it uses. A set of facial walks induced by a rotation system of $G$ uses each directed edge of $G$ exactly once by Lemma \ref{lemma:directedEdgesUsed}. There are $2m$ directed edges in $G$. \end{proof}

\begin{proposition} \label{prop:ijkcriteria}
    Let $G $ be a graph. Then any collection of facial walks $\mathcal{F}$ induced by a rotation system of $G$ is a set of closed trails that satisfies the following: \begin{enumerate}
            \item Every directed edge appears in exactly one closed trail.
            \item There do not exist two distinct closed trails $c_1$ and $c_2$ such that $c_1$ contains the sequence of directed edges $(a, b), (b, c)$ and $c_2$ contains the reversed sequence $(c, b), (b, a)$ for any vertex $b$ with $\deg(b) > 2$.
    \end{enumerate}
\end{proposition}
\begin{proof}
    Each facial walk is constructed by following directed edges according to the cyclic order at each vertex defined by the rotation system. Since the walk turns at each step without reversing the incoming edge, it is a closed trail. The rotation system specifies a unique next edge for every incoming edge, so the facial walks partition the $2m$ directed edges of $G$, with each used exactly once. Now suppose for contradiction that two facial walks $c_1$ and $c_2$ contain subpaths $(a,b),(b,c)$ and $(c,b),(b,a)$ respectively. Then $b$ must simultaneously follow both cyclic orders $(a,b,c)$ and $(c,b,a)$, which is impossible unless $\deg(b) = 2$, in which case the two orders are equivalent. 
\end{proof}

\subsection{Sufficiency and Limitations}
In the following, we discuss why the above necessary conditions are not sufficient for ensuring the correctness of \textsc{PAGE}. The following is a partial converse of Proposition \ref{prop:ijkcriteria}.
\vspace{-0.3em}
\begin{proposition} \label{prop:ijkcriteria_converse}
Let $G$ be a graph and $\mathcal{F}$ be a collection of closed trails satisfying the conditions of Proposition \ref{prop:ijkcriteria}. If every vertex $v$ of $G$ has $\deg(v)\leq 5$ then $\mathcal{F}$ is realizable as a set of facial walks of some rotation system of $G$. Additionally, if some vertex $v$ of $G$ has $\deg(v)>5$ the conditions in Proposition  \ref{prop:ijkcriteria} are not sufficient. 
\end{proposition}
\vspace{-1.2em}
\begin{proof}
Fix a vertex $v$ of degree $d$ and label its incident edges $e_1, \dots, e_d$. Traverse the closed trails in $\mathcal{F}$. Each time a closed trail enters $v$ along a directed edge $e_i = (u,v)$ and leaves along $e_j = (v,w)$, we say that the ordered pair $(e_i, e_j)$ \textit{occurs} at $v$. This defines a map
\vspace{-0.5em}
\begin{equation*}
\pi_v : \{e_1, \dots, e_d\} \to \{e_1, \dots, e_d\}, \qquad \pi_v(e_i) = e_j \quad \text{if } (e_i, e_j) \text{ occurs at } v.
\end{equation*}
We now show that $\pi_v$ is a cyclic permutation of the incident edges of $v$. Let $e_j = \{v,w\}$. Since each directed edge $(v,w)$ occurs in exactly one closed trail by Proposition~\ref{prop:ijkcriteria}, there exists some trail entering $v$ along $(u,v)$ and exiting along $(v,w)$. Let $e_i = \{u,v\}$. Then $(e_i, e_j)$ occurs at $v$, and so $\pi_v(e_i) = e_j$. Thus, $\pi_v$ is surjective. As the set $\{e_1, \dots, e_d\}$ is finite $\pi_v$ is also injective, so it is a bijective map from $\{e_1, \dots, e_d\}$ to itself and hence is a permutation of $\{e_1, \dots, e_d\}$. To show that $\pi_v$ has no fixed points, suppose for contradiction that $\pi_v(e_i) = e_i$. Then $(e_i, e_i)$ occurs at $v$, meaning some trail passes through $v$ along $(u,v)$ and immediately backtracks along $(v,u)$, contradicting the non-backtracking condition. To show that $\pi_v$ has no 2-cycles, suppose $\pi_v(e_i) = e_j$ and $\pi_v(e_j) = e_i$ for $e_i \ne e_j$. Let $e_i = \{u,v\}$ and $e_j = \{w,v\}$. Then there exist two trails $c_1$ and $c_2$ such that $c_1$ contains $(u,v), (v,w)$ and $c_2$ contains $(w,v), (v,u)$, violating condition (2) of Proposition \ref{prop:ijkcriteria}. Suppose that $\pi_v$ decomposes into $k > 1$ disjoint cycles $\pi_v = C_1 \circ C_2 \circ \dots \circ C_k$.
Since there are no fixed points or transpositions, each $C_i$ has length at least $3$, so
\begin{equation*}
d \geq \sum_{i=1}^k \text{Length}(C_i) \geq 3k \quad \Rightarrow \quad k \leq \left\lfloor \frac{d}{3} \right\rfloor.
\end{equation*}
But $d \leq 5$, so $k = 1$. Therefore, $\pi_v$ is a cyclic permutation. Defining $\pi_v$ at each vertex $v \in V$ gives a rotation system for $G$, and by construction, the facial walks of this rotation system are precisely the closed trails in $\mathcal{F}$. In Figure \ref{fig:degree6counterexample} the conditions of Prop. \ref{prop:ijkcriteria} are met, yet the closed trails cannot be realized by a rotation system.
\end{proof}

\section{Construction of the Algorithm}
\label{sec:Construction}
\subsection{Reducing the Search Space}
We reduce the search space by pruning all closed trail combinations that don't satisfy the necessary conditions to be realized by a rotation system of $G$. We consider collections that use each directed edge exactly once (Lem. \ref{lemma:directedEdgesUsed}), pass through each vertex according to its degree (Lem. \ref{lemma:vertexuses}), use exactly $2m$ directed edges (Lem.  \ref{lemma:cycledistribution}), and are realizable by a rotation system for vertex degrees less than $6$ (Prop. \ref{prop:ijkcriteria} and \ref{prop:ijkcriteria_converse}). The \texttt{PotentialMaxFit} procedure applies these lemmas to efficiently reject unrealizable closed trails early. Lemma \ref{lemma:facialsubsetcycles} indicates the need to filter out a subset of the closed trails. First, \texttt{PotentialMaxFit} discards sets with duplicate directed edges (Lem. \ref{lemma:directedEdgesUsed}). It tracks vertex usage (Lem. \ref{lemma:vertexuses}) and rejects overused vertices. Next, \texttt{PotentialMaxFit} ensures the sum of trail lengths is at most $2m$ (Lem. \ref{lemma:cycledistribution}) and rejects collections with conflicting subpaths (Prop. \ref{prop:ijkcriteria} and \ref{prop:ijkcriteria_converse}). These conditions allow us to prune unrealizable candidates early.

\subsection{Ensuring Sufficiency of the Conditions}
The final step ensures we only consider collections of closed trails realizable by a rotation system. We extend \texttt{PotentialMaxFit} to reject if the local rotation at any vertex splits into disjoint cycles (Prop. \ref{prop:ijkcriteria_converse}). This can be done efficiently by updating the partial cyclic orderings at each vertex each time we consider a new closed trail in our candidate set. This suffices, as cyclic orderings at each vertex define a rotation system.
\begin{figure}[h]
\centering
\includegraphics[width=0.4\linewidth]{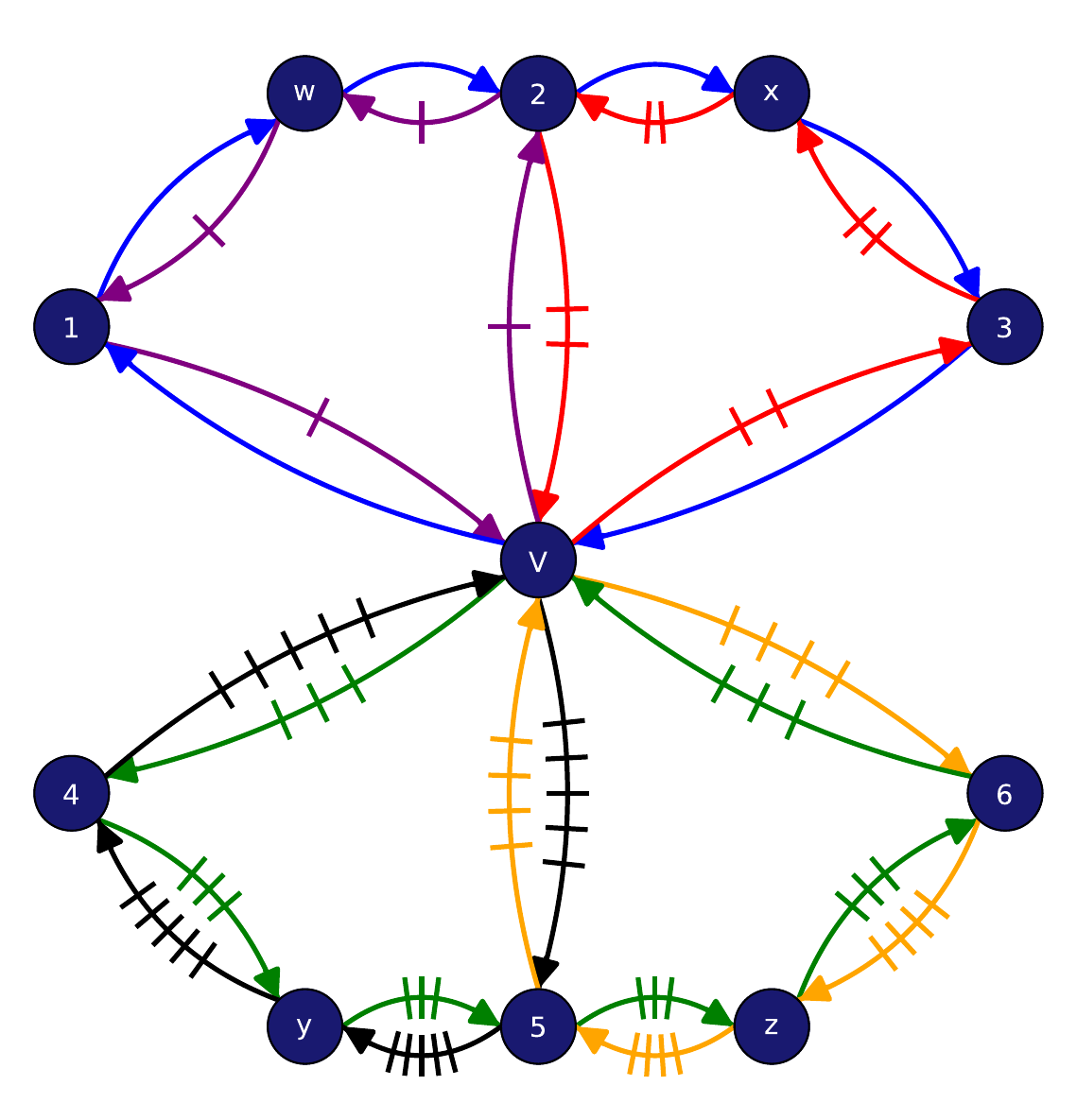}
\caption{Degree 6 example showing Proposition \ref{prop:ijkcriteria} is not sufficient.}
\label{fig:degree6counterexample}
\end{figure}

\subsection{Optimization via Backtracking}
At this point, we could finish a naive implementation of the algorithm by iterating through all closed trail combinations and rejecting any that don't satisfy \texttt{PotentialMaxFit}. Closed trails can be enumerated using Algorithm \ref{algo:cycle_finding} by going through each length $t, t+1, \ldots, 2m$. Yet, before implementing this we can still make some easy optimizations. We know that if $(a, b)$ occurs in some closed trail $c_1$ in a realizable set, then $(b, a)$ occurs in some other closed trail $c_2$ in the same set. We can thus avoid searching all closed trail combinations including $c_1$ if we show that \texttt{PotentialMaxFit} rejects all possible $c_2$ trails when $c_1$ is part of the set for any $(a, b)$. 

We do this using a recursive backtracking algorithm that picks an edge $e$ at each step that it knows must be part of a realizable set. It breaks ties by choosing in the way that minimizes the number of closed trails to search. When the candidate set is empty, this is the edge contained in the minimum number of closed trails so the search only has to consider a small number of starting trails which significantly reduces the branching factor at the top of the recursion tree and eliminates redundant exploration. When the candidate set is nonempty, this is the reverse edge of some edge in some closed trail in the set as long as the reverse edge is not in any closed trail in the set. We loop through all such reverse edges to choose the one contained in the minimum number of closed trails that have not been ruled out. We define \texttt{RequiredEdge} to be the function that computes this edge $e$ given a realizable set.

The backtracking algorithm backtracks when it has tried all the closed trails that contain the edge chosen at the current step and all of them are rejected by \texttt{PotentialMaxFit} because then the candidate set cannot be completed to a realizable set. The backtracking algorithm keeps adding closed trails to the candidate set until it reaches a full one with $2m$ edges by Lemma \ref{lemma:cycledistribution}. By construction of \texttt{PotentialMaxFit}, this is guaranteed to be realizable using a rotation system. In order to stop searching when we first find such a realizable candidate set, we iterate over trail combinations in order of decreasing number of trails. Since genus is inversely related to the number of faces in an embedding, and each face corresponds to a closed trail, ordering by decreasing number of trails naturally prioritizes embeddings of lower genus. This allows the algorithm to halt as soon as a minimal-genus realization is found, while still guaranteeing the search finds the minimum genus over all valid rotation systems without having to search all of them. To facilitate iterating over trail combinations, we define \texttt{AllClosedTrailsWithEdge} to enumerate all the closed trails containing a given edge. All this being said, we can finally formalize these optimizations using \texttt{PotentialMaxFit}, \texttt{RequiredEdge}, and \texttt{AllClosedTrailsWithEdge} in Algorithm \ref{algo:page} which completes \textsc{PAGE}.

\section{Formal Algorithms} 
\label{sec:algorithms}
\vspace{-1em}
\begin{algorithm}[H]
{\fontsize{8}{10.5}\selectfont
\setlength{\baselineskip}{10pt}
    \caption{\textsc{PAGE} - Calculate the Genus}
    \label{algo:page}
    \begin{algorithmic}[1]
    \setlength{\itemsep}{0pt}
\setlength{\parskip}{0pt}
        \Procedure{Search}{$G, \texttt{candidate\_set}$}
            \State $\texttt{required\_edge} \gets \textsc{RequiredEdge}(G, \texttt{candidate\_set})$
            \State $\texttt{trails\_to\_check} \gets \textsc{AllClosedTrailsWithEdge}(G, \texttt{required\_edge})$
            \For{each \texttt{trail} in \texttt{trails\_to\_check}}
                \State $\texttt{candidate\_set} \gets \texttt{candidate\_set} \cup \texttt{trail} $
                \If{\texttt{PotentialMaxFit}(\texttt{candidate\_set}) rejects}
                    \State $\texttt{candidate\_set} \gets \texttt{candidate\_set} \setminus \texttt{trail} $
                    \State \textbf{continue}
                \EndIf
                \If{\texttt{candidate\_set} contains $2m$ directed edges}
                    \State \textbf{return} $(V - E + |\texttt{candidate\_set}| - 2)/(-2)$
                \EndIf
                \State $\texttt{recurse} \gets \Call{Search}{G, \texttt{candidate\_set}}$
                \If{$\texttt{recurse} \neq -1$}
                    \State \textbf{return} \texttt{recurse} 
                    \Else
                    \State $\texttt{candidate\_set} \gets \texttt{candidate\_set} \setminus \texttt{trail} $
                    \State \textbf{continue}
                \EndIf
            \EndFor
            \State \Return -1
        \EndProcedure
        \State $\Call{Search}{G, \emptyset}$
    \end{algorithmic}}
\end{algorithm}
\subsection{Formal Pseudocode of PAGE}
This is the \textsc{PAGE} algorithm as detailed in \S\ref{sec:Construction}. On line 6, we reject all unrealizable candidate sets.
On line 10, we check if we have completed a realizable candidate set and if we do, we calculate the genus using Euler's formula and bubble up the result through the search tree on line 11.
On line 13, we keep adding more closed trails to the candidate set recursively.
On line 14 and 17, we reject candidate sets that become unrealizable for all possible trails we could add further down the search tree.
On line 15, we bubble up realizable candidate sets from further down the search tree if found.
On line 21, we've checked all closed trails and all of them reject so we backtrack. 

\begin{algorithm}
{\fontsize{8}{10.5}\selectfont
{ 
\setlength{\baselineskip}{10pt}
\caption{$k$-trail Finding Algorithm}
    \label{algo:cycle_finding}
    \begin{algorithmic}[1]
    \setlength{\itemsep}{0pt}
\setlength{\parskip}{0pt}
        \State \textbf{Input:} Graph $G = (V, E)$ and length $k$
        \State \textbf{Output:} Sequence of length-$k$ closed trails of $G$
        \State $\text{closed\_trails} \leftarrow \emptyset$
        \State $\text{queue} \leftarrow \text{FIFO queue with each single vertex path}$
        \While{$\text{queue}$ is not empty}
            \State $\text{path} \leftarrow \text{dequeue(queue)}$
            \If{$\text{len(path)} = k$}
                \If{first vertex of path is a neighbor of last vertex}
                    \If{last vertex is not the same as second vertex}
                        \If{directed edge from last to first vertex is not a repeat}
                            \State $\text{closed\_trails} \leftarrow \text{closed\_trails} \cup \text{path}$
                        \EndIf
                    \EndIf
                \EndIf
            \Else
                \For{each neighbor $v$ of last vertex in $\text{path}$}
                    \If{$\text{len(path)} > 2 \And $v is the second to last vertex of path}
                        \State \textbf{continue}
                    \EndIf
                    \If{directed edge from last vertex to $v$ is a repeat}
                        \State \textbf{continue}
                    \EndIf
                    \If{$v \geq \text{first vertex of path}$}
                        \State \text{enqueue(queue, path $\cup \ \{v\}$)}
                    \EndIf
                \EndFor
            \EndIf
        \EndWhile
    \end{algorithmic}
    }}
\end{algorithm}

\subsection{Formal Pseudocode of k-trail Finding Algorithm}

Since a realizable rotation system corresponds to a collection of closed trails whose lengths sum to $2m$ (Lemma~\ref{lemma:cycledistribution}), we need only generate closed trails of lengths that could feasibly appear in such a partition of the directed edge set. Moreover, if $t$ is the girth of $G$, then every non-backtracking closed directed trail has length at least $t$, since every such trail contains a cycle. Thus Algorithm~2 is only called for relevant lengths $k$ with $t\leq k\leq 2m$, and in practice one may further restrict to lengths that can occur in a partition of $2m$. We adapt the algorithm by Liu et al. \cite{liu2006} for enumerating simple cycles. It is advantageous for its simplicity, easily parallelized form, and ability to generate all cycles of a given length $k$ efficiently without keeping other cycles in memory or doing extensive computation on the graph beforehand. It is easily extended to closed trails as seen in the pseudocode. Line 8 ensures that the path forms a closed walk by checking that the last vertex is adjacent to the first. Line 9 prevents immediate reversals (backtracking) by ensuring that the last edge added does not reverse the previous one. Line 10 enforces that no directed edge is repeated within a trail. The condition in Line 23, requiring $v \geq$ the first vertex of the path, ensures that each cycle is only enumerated once up to rotation. We define the ordering on the vertices arbitrarily as long as its fixed for the given graph.

\section{Correctness and Output}
\begin{theorem} \textsc{PAGE} takes any graph $G$, calculates its genus, and produces the faces for an embedding of $G$ on a minimal genus surface $S$. \end{theorem}
\begin{proof} 
As shown in \S\ref{sec:TheoryPAGE} and \S\ref{sec:Construction}, \textsc{PAGE} discards sets of closed trails using necessary and sufficient criteria until it ends up with a realizable set of closed trails of maximum size. These closed trails are the facial walks which form the polygonal disc faces that when glued along shared edges and connected at shared vertices form $S$.
\end{proof}

This allows a ``proof certificate" to verify that the genus outputted by \textsc{PAGE} is no less than the minimum genus and is how we produced Figure \ref{fig:k33embedding} and the below color-coded faces of the minimum genus embedding of the Balaban (3, 10)-cage.

\begin{figure}[H]
    \centering
    \includegraphics[width=0.5\linewidth]{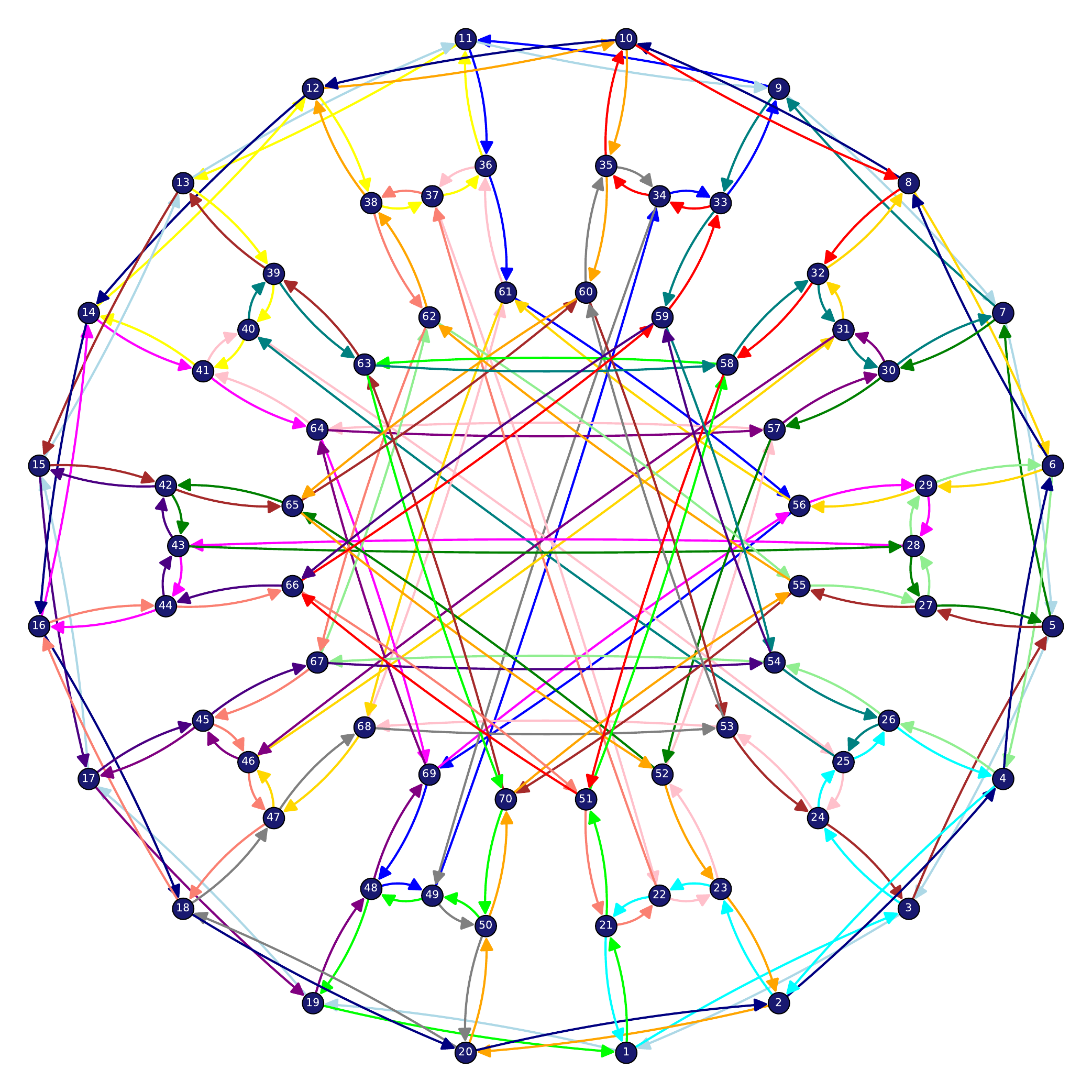}
    \caption{Genus 9 Embedding of the Balaban 10 Cage.}
\end{figure}

\begin{theorem}
    \textsc{PAGE} runs in $\mathcal{O}(n(4^m/n)^{n/t})$ time for graphs of girth $t$.
\end{theorem}
\begin{proof}
    The algorithm has some optimizations that allow it to stop early not captured in the below analysis, but it still holds as an upper bound on the runtime: the number of directed trails of any length is $\mathcal{O}(4^m)$ because it is the sum of the values in the $2m$-th row of Pascal's triangle. This can also be used to upper bound the number of closed trails which is what Algorithm \ref{algo:cycle_finding} enumerates. Let $c$ be the number of trails enumerated by our trail finding algorithm. Organizing the $c$ trails by vertex is $\mathcal{O}(2mc) < \mathcal{O}(2m4^m)$ since each trail has at most $2m$ edges. Choosing the required edge at each step that minimizes the branching factor can be done in $\mathcal{O}(n)$ by iterating through the vertices to find the most used ones. Looking up the trails that use the edge can be done in $\mathcal{O}(1)$ time via hashset lookup by the endpoints of the edge. Checking if a trail is used can be done in $\mathcal{O}(1)$ via hashset lookup. Checking if the edges of a trail are used is $\mathcal{O}(m)$ by storing the edges used in a hashset. Checking if adding a trail makes the rotation system unrealizable can be done in $\mathcal{O}(m)$ by storing the current rotation and using \texttt{PotentialMaxFit}. Let $f$ be the number of closed trails needed to have $2m$ edges in the candidate set, $b$ the number of trails by vertex, $u \leq b$ the number of unused trails by vertex, $a \leq u$ the number of unused trails with edges available, and $d \leq a$ the ones where the candidate set is realizable. Then each \texttt{Search} iteration satisfies $T(f) = \mathcal{O}(n) + \mathcal{O}(1) + b \mathcal{O}(1) + u \mathcal{O}(m) + a \mathcal{O}(m) + d T(f-1)$ and $T(0) = 0$. This implies $T(f) = \mathcal{O}(d^f \cdot (n + b + um + am)) = \mathcal{O}(d^f \cdot (n + b + m(u + a))) < \mathcal{O}(b^{f + 2})$. Let $h < 4^m$ be the number of start trails to try out. Then the total time of all \texttt{Search} iterations is $\mathcal{O}(b^{f + 2} \cdot h) < \mathcal{O}((4^m / n)^{n/t} \cdot 4^m) = \mathcal{O}(n(4^m/n)^{n/t})$.
\end{proof}
\subsection{Proof of Theorem \ref{theorem:4}} \label{proof:4}
When feasible, we search in decreasing order on the number of closed trails, equivalently increasing genus, so we can stop as soon as a realizable candidate set is achieved. Searching in this order may require too much computation before reaching a realizable candidate set, especially if there are too many closed trails to iterate through to identify the greatest-cardinality candidate sets as in \S4.4. In this case, we can instead iterate in increasing order of maximum closed-trail length. More precisely, for $\ell=t,t+1,\ldots,2m$, we search using only closed trails of length at most $\ell$. If no realizable candidate set exists using only trails of length at most $\ell-1$, then any remaining realizable candidate set must contain at least one closed trail of length at least $\ell$. Since every closed trail has length at least $t$, such a candidate set has cardinality
\[
F\leq 1+\left\lfloor \frac{2m-\ell}{t}\right\rfloor.
\]
By Euler's formula, this gives the lower bound
\[
g\geq 
\left\lceil
\frac{m-n+2-\left(1+\left\lfloor \frac{2m-\ell}{t}\right\rfloor\right)}{2}
\right\rceil .
\]
This produces a non-decreasing lower-bound sequence for the genus which eventually converges to $g$ in finitely many iterations. We now prove Theorem~6.
\begin{proof}
Proof. By Euler's formula [9], if $F$ is the number of facial walks, then
\[
g=\frac{m-n+2-F}{2}.
\]
Since every facial walk has length at least the girth $t$ and the facial walks together use all $2m$ directed edges, we have $tF\leq 2m$. Hence
\[
g\geq \frac{m-n+2-\frac{2m}{t}}{2}
=
\frac{2t-nt+m(t-2)}{2t}.
\]
Thus we may take the initial lower bound to be
\[
g_0=
\max\left\{
0,
\left\lceil
\frac{2t-nt+m(t-2)}{2t}
\right\rceil
\right\}.
\]
For the initial upper bound, the maximum possible orientable genus is bounded above by the one-face Euler characteristic bound, so we may take
\[
G_0=\left\lfloor\frac{m-n+1}{2}\right\rfloor.
\]
Thus $g_0\leq g\leq G_0$.

While \textsc{PAGE} searches it keeps the current bounds $g_{\text{current}}\leq g\leq G_{\text{current}}$. As \textsc{PAGE} iterates it may perform one of two updates,
\begin{itemize}
    \item (\textit{New Upper Bound}) When a candidate set of closed trails is formed containing all $2m$ directed edges, the previous call to \textsc{PotentialMaxFit} has already verified that it is realizable as a rotation system, so that each closed trail can be regarded as a facial walk of some rotation system. Then \textsc{PAGE} calculates the genus of the surface this rotation system embeds $G$ on, yielding $G_{\text{new}}=\frac{1}{2}\left(m-n+2-F\right)$ where $F$ is the amount of facial walks. If $G_{\text{new}} < G_{\text{current}}$ then we have a new upper bound and we append it to the upper-bound sequence $G_{s+1}:=G_{\text{new}}$. 
    \item (\textit{New Lower Bound}) $g_\text{current} = g_i$ follows directly from the iteration $i$ as previously defined.
\end{itemize}
Each update preserves $g_{\text{current}}\leq g \leq G_{\text{current}}$ and we have already shown that $g_\text{current}$ converges to $g$ in a finite number of steps. The same is true for $G_\text{current}$, since in $\mathcal{O}(m)$ iterations we will also have explored all closed trail combinations and by the correctness of \textsc{PAGE} have found a realizable candidate set with genus $g$. Thus \textsc{PAGE} halts with $g_i = g = G_j$ in a finite number of steps. 
\end{proof}

\section{Remarks and Comments}
\subsection{Extensions} 
As shown in various examples throughout this paper, determining the genus of a given graph is a longstanding problem in graph theory. Historically, the process of determining a graph's genus has often required extensive research and time, specifically focused on the individual graph in question. \textsc{PAGE} offers a new, practical, and fast method for calculating the genus of any graph. \textsc{PAGE} is also easily amenable to further optimization when specific information about a graph is known. For example, for large graphs, \textsc{PAGE} could be integrated with a computer algebra system to automate optimizations based on the graph's automorphism group. For instance, a graph like the Tutte-Coxeter (3,8) cage has the property that any path of up to 5 edges is equivalent to any other path up to automorphism \cite{Coxeter1958}. \textsc{PAGE} could leverage properties like this to select larger initial segments of closed trails while searching and even further reduce the search space if needed. Additionally, \textsc{PAGE} has the potential to answer many open conjectures in graph theory and advance the problem of completing the list of forbidden toroidal minors and indeed the sets of forbidden minors for surfaces of higher genus \cite{myrvold2018}. \textsc{PAGE} runs on large enough graphs to be useful for a number of applications such as designing circuit boards and microprocessors, roads and railway tracks, irrigation canals and waterways, quantum physics, and more.

\subsection{Runtime comparisons} The purpose of this section is to do an experimental comparison of the runtime of \textsc{PAGE} with \textsc{SAGEMath} and \textsc{multi\_genus} when computing the genus of the 3-regular Cage graphs (Sec. \ref{sec:examples_cages}) and the complete and complete-bipartite graphs (Sec. \ref{sec:prelim}). We can see that \textsc{PAGE} takes advantage of the high girth of the Cage graphs to scale exceptionally well. \textsc{PAGE} is also competitive on the complete and complete-bipartite graphs even though they have low fixed girths of 3 and 4 respectively.

\begin{table}[H]
\centering
\footnotesize
\begin{tabular}{|l|c|c|c|c|c|c|c|}
\hline
 k   & g   & v    & e     & genus    & \textsc{PAGE} (s) & \textsc{SAGEMath} (s) & \textsc{multi\_genus} (s)  \\
\hline
 3   & 3   & 4    & 6     & 0        & 0.008      & 0.004      & 0.006 \\
 3   & 4   & 6    & 9     & 1        & 0.008      & 0.039      & 0.006 \\
 3   & 5   & 10   & 15    & 1        & 0.008      & 0.027      & 0.006 \\
 3   & 6   & 14   & 21    & 1        & 0.008      & 0.010      & 0.006 \\
 3   & 7   & 24   & 36    & 2        & 0.010      & 1.737      & 0.006 \\
 3   & 8   & 30   & 45    & 4        & 0.032      & 118.958    & 0.012 \\
 3   & 9   & 58   & 87    & 7        & 1.625      & DNF        & 45.099 \\
 3   & 10  & 70   & 105   & 9        & 39.211     & DNF        & 9863.72 \\
 3   & 12  & 126  & 189   & 17       & 254.45     & DNF        & DNF \\
\hline
\end{tabular}
\caption{Genus and time measurements (in seconds) for Cage graphs.}
\label{tab:genus_data}
\end{table}

\begin{table}[H]
\centering
\footnotesize
\begin{tabular}{|l|c|c|c|c|c|c|c|}
\hline
 k   & g   & v    & e     & genus    & \textsc{PAGE} (s) & \textsc{SAGEMath} (s) & \textsc{multi\_genus} (s)  \\
\hline
2	 & 3   & 3	  & 3	  & 0	     & 0.005	         & 0.004	& 0.008 \\
3	 & 3   & 4	  & 6	  & 0	     & 0.005	         & 0.003	& 0.008 \\
4	 & 3   & 5	  & 10    & 1	     & 0.005	         & 0.005	& 0.008 \\
5	 & 3   & 6	  & 15	  & 1	     & 0.005	         & 0.023	& 0.008 \\
6	 & 3   & 7	  & 21	  & 1	     & 0.005	         & DNF	    & 0.009 \\
7	 & 3   & 8	  & 28	  & 2	     & 2.219	         & DNF	    & 0.008 \\
\hline
\end{tabular}
\caption{Genus and time measurements (in seconds) for complete graphs.}
\label{tab:genus_data_Kn}
\end{table}

\begin{table}[H]
\centering
\footnotesize
\begin{tabular}{|l|c|c|c|c|c|c|c|}
\hline
 k   & g   & v    & e     & genus    & \textsc{PAGE} (s) & \textsc{SAGEMath} (s) & \textsc{multi\_genus} (s)  \\
\hline
3	 & 4   & 6	  & 9	  & 1	     & 0.005	         & 0.047	& 0.006 \\
4	 & 4   & 8	  & 16	  & 1	     & 0.005	         & 0.010	& 0.010 \\
5	 & 4   & 10	  & 25	  & 3	     & 0.014	         & DNF	    & 0.008 \\
6	 & 4   & 12	  & 36	  & 4	     & 0.011	         & DNF	    & 0.009 \\
\hline
\end{tabular}
\caption{Genus and time measurements (in seconds) for complete bipartite graphs.}
\label{tab:genus_data_Knn}
\end{table}

\bibliographystyle{amsplain}
\bibliography{bibliography.bib}

\end{document}